\documentclass[a4paper,12pt]{article}

\usepackage{graphicx}
\usepackage{amsfonts,amsmath,amsthm,amssymb}
\usepackage{multirow}
\usepackage{booktabs}
\usepackage{tikz}
\usepackage{changes} 
\usepackage{mathtools}
\usepackage[authoryear]{natbib}
\usepackage{makecell} 
\usepackage{blindtext}
\usepackage{epigraph} 

\usepackage{fullpage}

\usepackage{footnote}
\usepackage{cancel}

\usepackage{setspace}

\def\mR{\mathbb{R}}
\def\mN{\mathbb{N}}

\def\R{\Rightarrow}

\def\LR{\Leftrightarrow}

\newtheorem{defi}{Definition}[section]
\newtheorem{propo}{Proposition}[section]
\newtheorem{lema}{Lemma}[section]
\newtheorem{coro}{Corollary}[section]
\newtheorem{ex}{Example}
\newtheorem{remark}{Remark}[section]
\def\N{\mathbb{N}}

\def\implica{\Rightarrow}

\def\pref{\succeq}
\def\prefs{\succ}

\def\implica{\Rightarrow}

\large

\newcommand{\pg}{p_{\gamma}}

\definecolor{gray}{rgb}{0.75,0.75,0.75}

\definecolor{MyDarkBlue}{rgb}{.1,0,.2}
\definecolor{MyBlue}{rgb}{.1,.1,.5}
\definecolor{MyBlue2}{rgb}{0.00,0.50,1.00}
\definecolor{MyBlue3}{rgb}{0.00,0.25,0.50}
\definecolor{MyBlue4}{rgb}{0.00,0.00,0.82}
\definecolor{MyBlack}{rgb}{0.00,0.00,0.00}
\definecolor{MyDarkBlue}{rgb}{0,0,1}
\definecolor{MyPink}{rgb}{0.81,0.06,0.46}
\definecolor{MyDarkGreen}{rgb}{0.00,0.50,0.25}
\definecolor{MyGreen}{rgb}{0,.4,0} 
\definecolor{MyGreen2}{rgb}{0.00,0.50,0.00} 
\definecolor{Red}{rgb}{1.00,0.00,0.00}{}
\definecolor{MyRed}{rgb}{.5,0,.5} 
\definecolor{MyBlack}{rgb}{0.00,0.00,0.00}
\definecolor{MyPur}{rgb}{.5,0,.25}

\def\@biblabel#1{\@ifnotempty{#1}{#1}}

\usepackage[format=hang,labelsep=colon,labelfont=sc]{caption}

\begin{document}


\title{\textbf{
Future-blindness and the product topology
}\thanks{
\textbf{Corresponding author}. Email: lorenzo.bastianello@unive.it } }
\author{ Marcel  Andrade \\
FGV-RJ\\
\and Lorenzo Bastianello  \\
Universit\`a Ca' Foscari of Venice\\
Universit\'e Panth\'eon-Assas \\
\and Jaime Orrillo  \\
Catholic University of Bras\'ilia
}
\date{
This version: January 2026
}

\maketitle

\begin{abstract}

We study future-blind preferences, which are preferences that heavily discount the future, within the space of infinite consumption streams. We give two definitions: $N$-blindness, where agents ignore periods beyond a fixed date $N$, and eventual blindness, where all but finitely many dates are neglected. Using a topological approach, we show that the finest topology ensuring eventual blindness coincides with the product topology. This provides a behavioral foundation for continuity in the product topology, which was considered for studying equilibrium existence in infinite-dimensional spaces. Finally, we characterize the dual spaces under these topologies.

\hfill \break \break \vspace{ -.5 mm }\textbf{Keywords: }Myopic Agents; Product Topology; Locally Convex Topologies; Non-Hausdorff Topologies.\newline

\end{abstract}

\newpage

\epigraph{``Apr\`es moi, le d\'eluge." }{\textit{Louis XV (1710-1774)}}

\section{Introduction}

The evaluation of infinite utility or consumption streams is a central problem in social choice theory, intergenerational welfare, and infinite-horizon economic models.  Since \citeauthor{diamond1965evaluation}'s [\citeyear{diamond1965evaluation}] seminal impossibility result, it has been well understood that natural requirements such as continuity, Pareto efficiency, and intergenerational equity cannot be simultaneously satisfied when the horizon is infinite. Intergenerational equity, also called anonymity, is an ethically appealing property that requires interchanging the income of two  generations should not alter the evaluation of the stream. Much of the literature has focused on how to reconcile equity with the two other properties, see \cite{pivato2024intergenerational} for an extensive survey on this topic. As already noted in  \cite{diamond1965evaluation}, equity can be interpreted as a form of extreme patience: the timing of consumption does not affect agents' preferences.

This paper takes the opposite viewpoint and studies a class of preference relations that we call \emph{future-blind}. Informally, future-blind preferences exhibit extreme forms of impatience: sufficiently distant generations are eventually ignored, even if large compensations are offered. While such behavior may be ethically controversial, it remains relevant. One can think of  infinitely lived but boundedly rational agents who maximize a ``short run utility”, as in \cite{lovo2010myopia}. Alternatively, agents may rationally discard all future dates beyond some point in time as it is too costly to consider all future periods (rational inattention). Finally, some decision maker may not care at all about future generations (the quote by Luis XV serves as an example).

We introduce two notions of future-blindness. The first one, \textit{$N$-blindness}, describes preferences of an agent who neglects all future periods starting from an exogenously given date $N$. The second more flexible notion, \textit{eventual blindness}, requires that there exists some (endogenous) finite horizon beyond which changes to the stream no longer affect strict preferences. This finite date exists even if the agent is compensated at a growing rate for receiving outcomes later in the future.

We study future-blindness following a topological approach.  Continuity of preferences is often seen as a mathematical requirement and it is interpreted as a robustness assumption. Intuitively, if preferences are continuous,  the ranking between streams should not change if the streams are slightly perturbed. It is well-known that in infinite dimensional spaces, such as $l^\infty$, continuity brings much more structure. As \cite{svensson1980equity} put it:
\begin{quote}
``A continuity assumption of preferences is not only a mathematical assumption, which is often the case in finite-dimensional spaces but also reflects a value judgement."
\end{quote}
Continuity with respect to different topologies entails specific discounting behavior. To give an intuition, consider the sup-norm topology over $l^\infty$ in which the distance between sequences ${\bf x}$ and ${\bf y}$ is given by $\sup_n |{\bf x}(n)-{\bf y}(n)|$. Consider the sequence (of sequences) $({\bf e}_i)_i$ where each ${\bf e}_i$ is equal to 1 in position $i$ and 0 otherwise. Note that the sup-norm-distance between ${\bf e}_i$  and the sequence ${\bf 0}$ is always 1, no matter how big the index $i$ is. Thus, continuity with respect to the sup-norm topology can be interpreted as a form of patience: one unit received in the distant future does count as one unit received closer to the present. This simple example shows that the sup-norm topology is not indicated to deal with future-blind preferences. Actually, \cite{lauwers1997continuity} argues that only topologies that are finer than the sup-norm topology can deal with equity/patience.\footnote{We underline that this is just an example and that continuity with respect to sup-norm topology is not equivalent to patience. For instance it is well known that the functional $\sum_t\beta^t u(x_t)$ defined on $l^\infty$ and axiomatized in \cite{koopmans1960stationary} is continuous with respect to the sup-norm topology.} The literature on (impossibility of) patience include the works of \cite{svensson1980equity}, \cite{campbell1985impossibility}, \cite{shinotsuka1997equity} \cite{sakai2003axiomatic}, \cite{banerjee2008continuity} and \cite{alcantud2014ordering}.

Our main objective is to clarify how these behavioral properties relate to continuity with respect to standard topologies on $l^\infty$, the space of bounded real-valued sequences. Building on the approach of \cite{brown1981myopic}, who provided a behavioral characterization of the Mackey topology via myopia,\footnote{Myopic agents discount the future too, but not as strongly as future-blind decision makers. A precise definition of myopia is given in Section~\ref{sec:Nblind}.} we ask the following question: which continuity assumptions correspond to the extreme forms of tail insensitivity implied by future-blindness? Our main results provide a sharp answer. First, we show that the largest topology under which all continuous preferences are $N$-blind is characterized by the fact that sequences differing from $0$ only after period $N$ are arbitrarily close to the ${\bf 0}$ sequence. Second, our main proposition shows that the finest locally convex topology under which all continuous preferences are eventually blind coincides with the product topology. This establishes a direct behavioral interpretation of product-topology continuity: it corresponds to welfare comparisons that eventually disregard all sufficiently remote generations, regardless of compensating improvements.

Our result parallels \cite{brown1981myopic} main theorem that shows that the myopic topology is equivalent to the Mackey topology. It is a significant result as the Mackey topology was used by \cite{bewley1972existence} to prove the existence of equilibria in the infinite dimensional space $L^1$. Proposition~\ref{prop:Tb=Tp} fills a gap in the literature: the product topology was used by \cite{peleg1970markets} to show the existence of equilibria when the commodity space is $l^\infty$. Therefore, we explicitly provide a behavioral foundation for continuity with respect to the product topology: it implies that agents are eventually blind. It is worth noticing that continuity with respect to the product topology is a very convenient assumption. In fact, budget sets are compact in the product topology (by Tychonoff theorem) and hence the maximization of a continuous function yields the existence of optimal solutions. See \cite{lauwers1997continuity} for a discussion about the efficiency of the product topology.

Finally, we study dual spaces when $l^\infty$ is paired with our different topologies. Dual spaces contain continuous linear functionals, which generalize the notion of prices in infinite dimensions. We show that the dual space of $l^\infty$ when paired with the $N$-blind topology is the space of sequences that are equal to zero for all $n> N$. The dual space of the eventually blind topology is the space of sequences that are zero for all but finitely many elements (of course, this result is expected, as the eventually blind topology  is equal to the product topology).

The rest of the paper is organized as follows. Section~\ref{sec:MathPrelim} explains notations and contains some reminders about the topologies that we use. Section~\ref{sec:Nblind} and Section~\ref{sec:eventual-blind} present the concepts of $N$-blindness and eventual blindness respectively, and study their related topologies. Section~\ref{sec:dual_applications} investigates dual spaces. Section~\ref{sec:conclusions} concludes.


\section{\label{sec:MathPrelim} Notation and Mathematical Preliminaries}

The set of  all natural numbers $\mathbb{N}=\{ 0, 1, \dots, \}$ represents time periods. A real-valued sequence is any function ${\bf x}:\mathbb{N} \to  \mathbb{R}$ whose values are denoted ${\bf x}(n)$ or $x_n$. We also write ${\bf x} = (x_0, x_1, \dots )$. Therefore $\mathbb{R}^{\mathbb{N}}$ is the vector space of all sequences. This paper deals with the subspace  $l^\infty$ of real-valued bounded sequences. More specifically, if we denote the \textit{sup-norm}  $\| {\bf x} \|_{\infty} = \sup_{n\in \mathbb{N}}|x_n|$, then
$$
l^\infty=\{{\bf x}\in \mR^\mN :\| {\bf x} \|_{\infty} < \infty\}.
$$
The set  $l^\infty$  is interpreted as the set of all possible income or consumption streams. The sets $l^\infty_+$ and $l^\infty_{++}$ denote the subsets of $l^\infty$ with non-negative and positive sequences, respectively.  A preference relation $\pref $ on  $ l^{\infty}$ is a complete and transitive binary relation representing the tastes of a decision maker or a social planner. Note that preferences are not indexed by time; this implies that the decision maker is only making decisions in period 0.

The sum between two sequences and the multiplication by a scalar are the pointwise sum and multiplication, i.e. for all ${\bf{x},\bf{y}}\in l^\infty$ and $\alpha\in \mR$, ${\bf{x}+\bf{y}}=(x_0+y_0,x_1+y_1,\dots)$ and $\alpha{\bf{x}}=(\alpha x_0,\alpha x_1,\dots)$. We denote ${\bf{e}}_i$ the sequence that is equal to 1 in position $i$ and 0 everywhere else, i.e. ${\bf{e}}_i(i)=1$ and ${\bf{e}}_i(k)=0$ for all $k\neq i$.

In the paper, we will deal with the following subspaces of $l^\infty$:
\begin{align*}
 c_{00}^N&=\{{\bf x}\in l^\infty : x_n=0 \text{ for all } n>N \}, \\
 c_{00}&=\{{\bf x}\in l^\infty : x_n=0 \text{ for all but finitely many } n \}, \\
 c_{0} &=\{{\bf x}\in l^\infty : \lim_{n} x_n =0 \}, \\
 l^1 &=\{{\bf x}\in l^\infty : \sum_{n} |x_n|  < \infty \}.
\end{align*}
We remark that $c_{00}^N \subsetneq c_{00} \subsetneq l^1 \subsetneq  c_{0} \subsetneq l^\infty$.

For all $n\in \mN$ we define the \textit{head operator} $H_n:l^\infty\rightarrow l^\infty$ and the  \textit{tail operator} $T_n:l^\infty\rightarrow l^\infty$ for all ${\bf x}\in l^\infty$ by
\[
\begin{array}{cccccccc}
    H_n({{\bf x}}) &= &(x_0,  & \dots, & x_{n}, & 0, & 0, & \dots ),\\
    T_n({{\bf x}}) &= &(0,  & \dots, & 0, & x_{n+1},  & x_{n+2},  & \dots ).\\
\end{array}
\]

We are interested in studying how continuity of preferences translates into time discounting. Let $\tau$ be a topology on $l^{\infty}.$  A preference is continuous with respect to $ \tau $ if  for all ${\bf y} \in l^{\infty}$ the sets $\{ {\bf x} \in l^{\infty} : {\bf x}\succ {\bf y} \} $ and $\{ {\bf x}  \in l^{\infty} : {\bf y} \succ {\bf x} \}$  are $\tau$-open. The rest of this section contains some reminders about topology. 

\medskip

\noindent \textit{Reminders.}

In this paper, we restrict our attention to locally convex topologies. A topological vector space is called \textit{locally convex} if it has a neighborhood basis at the origin consisting of convex sets. Usually, it is easier to work with seminorms. A \emph{seminorm} on $l^\infty$ is a function $p:l^\infty\rightarrow \mR$ such that for all ${\bf x},{\bf y} \in l^\infty$ and for all $ \alpha\in \mR$, $(i)$ $p({\bf x}+{\bf y})\leq p({\bf x})+p({\bf y})$ and $(ii)$ $p(\alpha {\bf x})=|\alpha|p({\bf x})$. If moreover $p({\bf x})=0$ if and only if ${\bf x}=0$, then $p$ is a \emph{norm}. 

A family of seminorms $\Gamma$ \emph{generates} a topology $\tau$ if $\tau$ is the smallest topology making all seminorms in $\Gamma$ continuous. More explicitly, if $\Gamma$ is a family of seminorms on $l^\infty$, a subset $ \mathcal{U} \subset l^{\infty}$ is said to be an open neighborhood of ${\bf 0}\in l^\infty$ if ${\bf 0}\in \mathcal{U} $ and if there are a finite family of seminorms $q_1,\dots,q_k\in \Gamma$ and a real number $\epsilon > 0$ such that 
$$
 \cap_{i=1}^k \{ {\bf x}\in l^\infty : q_i({\bf x}) < \epsilon \} \subset \mathcal{U} .
$$
Every topology generated by seminorms is a locally convex topology and every locally convex topology is generated by seminorms, see \cite{aliprantis06}, Theorem 5.73.

Suppose that the family of seminorms $\Gamma$ generates topology $\tau$. Then a sequence (of sequences) $({\bf x}_n)_n$ converges in $\tau$ to a sequence ${\bf x}\in l^\infty$, denoted ${\bf x}_{n}\xrightarrow{\tau}_{n} {\bf x}$,  if and only if $\pg({\bf x}_{n}-{\bf x})\rightarrow_{n} 0$ for every $\gamma\in \Gamma $ (see \cite{aliprantis06}, Lemma 5.76). 

A locally convex topology  generated by $\Gamma$ has the \textit{Hausdorff  property} if $\cap_{q\in \Gamma } \{ {\bf x}\in l^{\infty} : q({\bf x}) =0  \} = \{ {\bf 0}\}.$ Let $\overline{A}$ denote the closure of a set $A$. Then a locally convex topology  generated by $\Gamma$ is Hausdorff if and only if $\overline {\{ {\bf 0} \} } = \{ {\bf 0}\}$, i.e. for all $ {\bf x} \not = {\bf 0}$, there exists $q\in \Gamma$ such that $q({\bf x}) \not = 0$.

Let $\tau_1$ and $\tau_2$ be two topologies on $l^{\infty}$. If $\tau_1\subseteq \tau_2$ we say that $\tau_1$ is \emph{weaker (or coarser)} than $\tau_2$ or that $\tau_2$ is \emph{stronger (or finer)} than $\tau_1$. If additionally $\tau_1\neq \tau_2$, we write $\tau_1\subsetneq \tau_2$ and we say that $\tau_1$ is strictly weaker (or strictly coarser) than $\tau_2$ or that $\tau_2$ is strictly stronger (or strictly finer) than $\tau_1$.

Given a topology $\tau$ on $l^{\infty}$, the \emph{(topological) dual} of $l^{\infty}$ with respect to $\tau$ is the set of $\tau$-continuous linear functionals on $l^{\infty}$ and it is denoted $(l^{\infty},\tau)'$.  In economics, dual spaces play an important role as prices of commodities lie in the dual.

There are several interesting and well-known topologies on $l^{\infty}$. In this paper, the ones listed below will play a central role.
\begin{itemize}
    \item The \textit{sup-norm topology}, denoted $\tau^\infty$, is generated by the sup-norm (already defined above). We denote $ba$ (resp. $ca$) the set of \textit{boundedly} (resp. \textit{countably}) \textit{additive charges} over the power set $2^\mN$. Then it is well known that the dual space $(l^{\infty},\tau^{\infty})'$ is isomorphic to $ba$,  i.e. every $\tau^{\infty}$-continuous linear functional over $l^{\infty}$ can be represented by $\int {\bf x} d\mu$ with $\mu\in ba$ (the integral is known as the Dunford integral, see Chapter 4 of \cite{rao1983theory}). 
  
    \item The \textit{product topology}, denoted $\tau^p$, is generated by the family of seminorms $\{p_n| n\in \mN\}$, where for all $n\in \mN$, $p_n({\bf x})=|x_n|$, for all ${\bf x}\in l^\infty$. It can be proved (see \cite{aliprantis06} Theorem 16.3), that the dual space $(l^{\infty},\tau^p)'$ is  isomorphic to $c_{00}$.
    
    \item The \textit{Mackey topology}, denoted  $\tau^{Mc}$, is the strongest topology in $l^\infty$ with $l^1$ as its topological dual. 

    \item The \textit{strict topology}, denoted  $\tau^{s}$,  introduced by \cite{buck1958bounded},  is generated by the family of seminorms  $q_{ \bf a}({\bf x })=\sup_n|a_nx_n|$, where ${\bf a } \in c_0$. By \cite{conway1967strict}, the strict topology coincides with the Mackey topology in $l^\infty$.
\end{itemize}
We remark that $\tau^p \subsetneq \tau^{Mc}=\tau^{s} \subsetneq  \tau^\infty$.



\section{$N$-Blindness}\label{sec:Nblind}

Section~\ref{sec:Nblind} analyzes decision makers who neglect the future starting from an exogenous time period $N\in \mN$. The aim is to find a suitable topology over $l^\infty$  that characterizes such behavior.

Recall that, for $N\in \mN$, the head and tail operators are defined respectively by $H_N({\bf x} )=(x_0,\dots,x_{N},0,\dots)$ and $T_N({\bf x} )=(0,\dots,0,x_{N+1},x_{N+2},\dots)$.  Consider the following definition. 

    \begin{defi}\label{def:Nblindpref1}
    Let  $N\in \mN$ be any fixed natural number.  A preference relation  $\pref$ on $l^{\infty}$ is $N$-blind if for all ${\bf x }, {\bf y }, {\bf z }\in l^{\infty}$, if ${\bf x }\prefs {\bf y }$ then ${\bf x }\prefs {\bf y } + T_n({\bf z})$ for all $n\geq N$.
    \end{defi}

Definition \ref{def:Nblindpref1} states that an agent is $N$-blind if, whenever she has a strict preference for stream ${\bf x }$ over stream ${\bf y }$, she will continue to strictly prefer ${\bf x }$ over ${\bf y }$ even after stream ${\bf z }$ is added to ${\bf y }$ starting from period $N$. The strict preference is maintained independently of how ``good" stream ${\bf z }$ is, such behavior indicates that the decision maker only cares about periods up to $N$.

We now give two equivalent definitions of $N$-blindness.

\begin{propo}\label{prop:Nblindpref}
Let  $N\in \mN$ be any natural number. Then the following assertions are equivalent.
\begin{itemize}
    \item[(i)]  $\pref$ is  $N$-blind.
    \item[(ii)] For all ${\bf x }, {\bf y }, {\bf z }, {\bf w } \in l^{\infty}$, if ${\bf x }\prefs {\bf y }$ then ${\bf x } + T_n({\bf w})\prefs {\bf y } + T_n({\bf z})$ for all $n\geq N$.
    \item[(iii)] For all $ {\bf x }, {\bf y }  \in l^\infty$, if $H_N({\bf x})=H_N({\bf y})$ then ${\bf x }\sim{\bf y }$.
    \item[(iv)] For all $ {\bf x }\in l^\infty, H_N({\bf x })\sim{\bf x } $.
\end{itemize}
\end{propo} 
\begin{proof}
$(i)\Rightarrow (ii)$. Suppose by contradiction that ${\bf x }\prefs {\bf y }$ but there are $\bar{{\bf  z}}, \bar{{\bf w} } \in l^{\infty}$ and $n\geq N$ such that ${\bf y } + T_n({\bf \bar{z}})\pref {\bf x } + T_n({\bf \bar{w}})$. Then we get ${\bf x} \prefs {\bf y } + T_n({\bf \bar{z}})\succeq {\bf x } + T_n({\bf \bar{w}})$ where the strict preference comes from $(i)$ and the weak one by the hypothesis. But then ${\bf x } \prefs {\bf x } + T_n({\bf \bar{w}})$ and taking ${\bf z}=-{\bf \bar{w}}$ by $N$-blindness we obtain 
${\bf x } \prefs {\bf x } + T_n({\bf \bar{w}})+T_n(-{\bf \bar{w}})={\bf x }$, a contradiction.

\medskip

$(ii)\Rightarrow (i)$. Obvious.

\medskip 

$(i)\Rightarrow (iii)$. Take  $ {\bf x }, {\bf y }$ such that $H_N({\bf x})=H_N({\bf y})$. Suppose that  ${\bf x }\succ {\bf y }$. If $T_N({\bf x})=T_N({\bf y})$, then ${\bf x}={\bf y}$ and ${\bf x }\succ {\bf y }$ yields a contradiction. If $T_N({\bf x})\neq T_N({\bf y})$, then taking ${\bf z}={\bf x}-{\bf y}$ in    $(i)$ we get ${\bf x }\succ {\bf y }+T_N({\bf z} )= {\bf y} + T_N({\bf x}-{\bf y})= {\bf y} + T_N({\bf x})-T_N({\bf y})= H_N({\bf y})+T_N({\bf y})+T_N({\bf x})-T_N({\bf y})= H_N({\bf y}) + T_N({\bf x}) = H_N({\bf x}) + T_N({\bf x})={\bf x}$, another contradiction. 

\medskip 

$(iii)\Rightarrow (i)$. Take  ${\bf x }, {\bf y }, {\bf z }\in l^{\infty}$ such that ${\bf x }\prefs {\bf y }$. Then since $H_N({\bf y})=H_N( {\bf y} + T_N({\bf z}))$, by point $(iii)$ we have  ${\bf y}\sim {\bf y} + T_N({\bf z})$. Therefore, $ {\bf x }\prefs {\bf y }+T_N({\bf z})$. 

$(iii)\Rightarrow (iv)$. Obvious.

$(iv)\Rightarrow (iii)$. Take ${\bf x }, {\bf y }\in l^{\infty}$ and suppose $H_N({\bf x})=H_N({\bf y})$. Then ${\bf x}\sim H_N({\bf x})=H_N({\bf y})\sim {\bf y}$.

\end{proof}

The result above gives alternative characterizations of $N$-blindness. Points $(iii)$ and $(iv)$ simply say that an $N$-blind decision maker is indifferent whenever sequences are equal until period $N$. Point $(ii)$ in Proposition \ref{prop:Nblindpref} says that a decision maker's strict preference  ${\bf x }\succ{\bf y }$ will not be changed not only if ${\bf y }$ is improved by ${\bf z}$ after period $N$ (as in Definition \ref{def:Nblindpref1}), but also if ${\bf x }$ is worsened through ${\bf w}$.

$N$-blindness is a very strong example of preferences for present over future consumption.  An $N$-blind agent only considers the first $N$ periods and neglects all the others. It is reminiscent of the (dynamic) model of \cite{lovo2010myopia} in which agents are infinitely lived but act in each period as if they would live only $N$ more periods. For comparison, consider the notion of myopia\footnote{\cite{brown1981myopic} considered also a notion of weak myopia, in which ${\bf z}$ is the constant unit sequence. If preferences are monotonic, these two concepts coincide.} introduced by \cite{brown1981myopic}.

 \begin{defi}[\cite{brown1981myopic}]\label{def:myopiaBL}
A preference relation $\pref$ on $l^{\infty}$ is myopic if  for all ${\bf x }, {\bf y }, {\bf z }\in l^{\infty}$, there exists $N\in \mN$ such that if ${\bf x }\prefs {\bf y }$ then ${\bf x }\prefs {\bf y } + T_n({\bf z})$ for all $n\geq N$. 
\end{defi}
Since in the definition of myopia $N$ depends on sequences ${\bf x }, {\bf y }$ and $ {\bf z }$ one can readily see that $N$-blindness is a severe form of myopia, where the period $N$ in which the agent starts neglecting the future is fixed and does not depend on the streams under consideration. Hence, for all  $N\in \mN$, an $N$-blind preference is also myopic.

\cite{brown1981myopic} introduced the definition of myopic topologies.

 \begin{defi}[\cite{brown1981myopic}]
A locally convex topology $\tau$ is myopic if every $\tau$-continuous preference is myopic. 
\end{defi}

Likewise, we introduce the following definition of  $N$-blind topologies.

 \begin{defi}
A locally convex topology $\tau$ is $N$-blind if every $\tau$-continuous preference is $N$-blind.
\end{defi}

As it turns out, \cite{brown1981myopic} show that a locally convex topology generated by a family of seminorms $\Gamma$ is myopic if and only if $p(T_n({\bf x }))\rightarrow_n 0$ for all $p\in \Gamma$ and all ${\bf x }\in l^\infty$. A seminorm with this property will be called a \textit{myopic seminorm}. The following result characterizes $N$-blindness in locally convex topologies through properties of seminorms.

\begin{propo}\label{prop:Nblind-seminorm}
Let  $\tau$ be a locally convex topology on $l^{\infty}$ generated by a family of seminorm $\Gamma$. Then  $\tau$ is $N$-blind  if and only if for all $p\in \Gamma$, $p(T_n({\bf x }))=0$ for all $n\geq N$ and all ${\bf x }\in l^\infty$.
\end{propo}

\begin{proof}
$\Rightarrow$ We follow the ideas in the proof of Lemma 1 in \cite{brown1981myopic}. Suppose that $\tau$ is generated by the family of seminorms $\Gamma=\{ p_\alpha: \alpha\in A\}$. Then for all $\alpha\in A$ we can define a $\tau$-continuous preference $\pref_\alpha$ on $l^\infty$ by ${\bf x }\pref_\alpha {\bf y }$ if and only if $p_\alpha({\bf x })\geq p_\alpha({\bf y })$. \\
Fix now $p_\alpha\in\Gamma$. If $p_\alpha=0$ we are done. Otherwise there is $\bar{{\bf x }}$ such that $p_\alpha(\bar{{\bf x }})>0$. Then for all $k\in \mathbb{N}$, $\frac{1}{k}p_\alpha(\bar{{\bf x }})=p_\alpha\left(\frac{1}{k}\bar{{\bf x }}\right)>0=p_\alpha({\bf 0 })$ and therefore $\frac{1}{k}\bar{{\bf x }}\succ_\alpha {\bf 0 }$. Since $\tau$ is $N$-blind and $\pref_\alpha$ is $\tau$-continuous, for all $n\geq N$ and for all ${\bf z }\in l^\infty$ $\frac{1}{k}\bar{{\bf x }}\succ_\alpha T_n({\bf z })$. Hence $\frac{1}{k}p_\alpha(\bar{{\bf x }})> p_\alpha(T_n({\bf z }))$ for all $k\in\N$. We conclude that $p_\alpha(T_n({\bf z }))=0$ for all $n\geq N$ and all ${\bf z }\in l^\infty$.

\medskip

$\Leftarrow$ Suppose that $p\in \Gamma$, $p(T_n({\bf x }))=0$ for all $n\geq N$ and all ${\bf x }\in l^\infty$ and that $\pref$ is $\tau$-continuous. We will use point $(iii)$ of Proposition \ref{prop:Nblindpref}. Take ${\bf x },{\bf y } \in l^\infty$ such that $H_N({\bf x })=H_N({\bf y })$ but suppose by contradiction that ${\bf x }\succ{\bf y }$. Since $\pref$ is $\tau$-continuous the set $\mathcal{U}=\{{\bf z }\in l^\infty : {\bf x }\succ{\bf z }\}$ is open and ${\bf y}\in \mathcal{U}$. Since for all $p\in \Gamma$, $p({\bf y }-({\bf y} + T_N({\bf z }))) = p(T_N({\bf z }))=0$ for all ${\bf z }\in l^\infty$, we have that ${\bf y} + T_N({\bf z })\in \mathcal{U}$ for all ${\bf z }\in l^\infty$. But then taking ${\bf z }= {\bf x }-{\bf y }$ we obtain that ${\bf x }\in \mathcal{U}$, a contradiction. Therefore ${\bf x }\sim {\bf y }$.
\end{proof}

Proposition \ref{prop:Nblind-seminorm} shows that seminorms generating $N$-blind topologies assign value zero to sequences that are zero up to period $N$. These seminorms will be called \textit{$N$-blind seminorms}. Clearly, every $N$-blind seminorm is also a myopic seminorm. An example of an $N$-blind seminorm is the following.

\begin{ex}\label{ex:N-blind-seminorm}
  Note that if $k\leq N-1$, the seminorm $p_k$ defined for all ${\bf x } \in l^\infty$ by
    $$
    p_k({\bf x })=|x_k|
    $$
    is a $N$-blind seminorm by Proposition \ref{prop:Nblind-seminorm}.
\end{ex}

Since continuity of preferences is defined in terms of open upper and lower contour sets, the finer a topology, the ``easier"  for a preference relation to be continuous with respect to it. We would like therefore to consider the finest topology with the property of being $N$-blind. Consider the following definition.

\begin{defi}\label{def:tau_N_blind}
    We denote $\tau^b_N$ the finest $N$-blind locally convex topology over $l^\infty$. The family of seminorms generating it is denoted $\Gamma^b_N$.
\end{defi}

We will prove now that the topology  $\tau^b_N$ exists for all $N\in \N$ and we will give an explicit characterization of all seminorms in $\Gamma^b_N$. 

Recall that two families of seminorms are said to be equivalent if they generate the same topology. Consider two families of seminorms $\{p_\alpha: \alpha\in A\}$ and  $\{q_\beta: \beta\in B\}$. They are equivalent if and only if every seminorm $p_\alpha$ (resp. $q_\beta$) is \emph{dominated} by a finite number of seminorms in $\{q_\beta: \beta\in B\}$ (resp. $\{p_\alpha: \alpha\in A\}$). Explicitly, they are are equivalent if and only if $(i)$ for all $\alpha\in A$ there exist $\beta_1,\dots, \beta_{n_\alpha}\in B$ and $c>0$ such that $p_\alpha({\bf x })\leq c(q_{\beta_1}({\bf x })+\dots+q_{\beta_{n_\alpha}}({\bf x }))$ for all ${\bf x } \in l^\infty$; and $(ii)$ for all $\beta\in B$ there exist $\alpha_1,\dots, \alpha_{n_\beta}\in A$ and $d>0$ such that $q_\beta({\bf x })\leq d(p_{\alpha_1}({\bf x })+\dots+p_{\alpha_{n_\beta}}({\bf x }))$ for all ${\bf x } \in l^\infty$. 
We are ready to give a complete characterization of seminorms generating the topology $\tau^b_N$.

\begin{propo}\label{prop:N-blind_topo}
$\tau^b_N$ exists for all $N\in \N$,  and it is generated by  the family of seminorms $P^N=\{p_k:k=0,\dots, N\}$, where $p_k({\bf x })=|x_k|$ for all ${\bf x } \in l^\infty$.
\end{propo}

\begin{proof}
Consider the family of  $\Gamma^b_N$ containing all seminorms $p$ such that $p(T_n({\bf x }))=0$ for all $n\geq N$ and all ${\bf x }\in l^\infty$. Example \ref{ex:N-blind-seminorm} shows that, if $k\leq N-1$, then $p_k({\bf x })=|x_k|$ belongs to  $\Gamma^b_N$ . Therefore this family is non empty. Hence $\Gamma^b_N$ generates a locally convex topology that is $N$-blind by Proposition~\ref{prop:Nblind-seminorm}. Moreover, it is the finest as it contains all such seminorms.\\
We show now that $\Gamma^b_N$ and $P^N$ generate the same topology.  By Proposition~\ref{prop:Nblind-seminorm}, every seminorm in $P^N$ is a $N$-blind seminorm. Therefore $P^N\subseteq \Gamma^b_N$.
Now, let $q\in \Gamma^b_N$ and define $$c:= \max_{i=0,...,N} q({\bf e}_i)$$
where ${\bf e}_i$ denotes the sequence equal to 0 everywhere except in position $i$, where it is equal to 1. 
By Proposition~\ref{prop:Nblind-seminorm}, $q(T_N({\bf x}))=0$. Therefore, if $c=0$, then $q=0$  and $q({\bf x})\leq p_i({\bf x})$ for all ${\bf x}\in l^\infty$, and $i=0,...,N$. On the other hand if $c>0$, notice that 
$$
q({\bf x})=q(H_N({\bf x})+T_N({\bf x}))\leq q(H_N({\bf x}))+q(T_N({\bf x}))=q(H_N({\bf x})).
$$
Since $H_N({\bf x})=\sum_{i=0}^N x_i{\bf e}_i$, we have 
$$
q({\bf x }) \leq   q\left(\sum_{i=0}^N x_i{\bf e}_i \right)
                \leq \sum_{i=0}^N  q( x_i{\bf e}_i) 
                = \sum_{i=0}^N |x_i|q({\bf e}_i)
                \leq \max_{i}q({\bf e}_i)\sum_{i=0}^N |x_i| 
                \leq c \sum_{i=0}^n p_i({\bf x }).
$$
Therefore, all the seminorms in $\Gamma^b_N$ are dominated by seminorms in $P^N$ and we can conclude that $\Gamma^b_N$ and $P^N$ generate the same topology.
\end{proof}

The following two corollaries show respectively that $\tau^{b}_N$ is not Hausdorff and that $(\tau^{b}_N)_{N\in \N}$ forms a nested sequence of topologies.

\begin{coro}\label{cor:non-Haussdorf}
$\tau^{b}_N$ is not Hausdorff and the closure of $\{ \bf 0\}$ respect to the topology $\tau^{b}_N$ is $T_{N}=\{T_N({\bf x}): {\bf x}\in l^\infty\}.$
\end{coro}
\begin{proof}
Note that for all ${\bf x}\in l^\infty$, $T_N({\bf x})$ belongs to all neighborhoods of $\{ \bf 0\}$ and therefore  $T_N\subseteq  \overline{\{{ \bf 0}\}}$. On the other hand if ${\bf x}\not\in T_N$, then there exists a seminorm $p$ in $P^N$ such that $p({\bf x})>0$, which implies that ${\bf x}\not\in \overline{\{{ \bf 0}\}}$. Hence $\overline{\{{ \bf 0}\}}\subseteq T_N$ and we obtain $\overline{\{{ \bf 0}\}}= T_N$. This result implies that $\tau^{b}_N$ is not Hausdorff.
\end{proof}

\begin{coro}\label{cor:nested}
 For any  pair  $M,N \in \mN$ the following holds:
    $$M <  N \implica \tau^{b}_M \subsetneq \tau^{b}_N \subsetneq \tau^p.$$
\end{coro}
\begin{proof}
By Proposition \ref{prop:N-blind_topo}, $\tau^b_N$  is generated by  the family of seminorms $P^N=\{p_k:k=0,\dots, N\}$. Clearly,  if $M <  N$ then $P^M\subsetneq P^N$ and hence $\tau^{b}_M \subsetneq \tau^{b}_N$. Moreover, since  $\cup_{N\in \N} P^N$ is the set of seminorms generating the product topology, for all $N$, $ \tau^{b}_N \subsetneq \tau^p$.
\end{proof}

Topology $\tau^b_N$ of Definition \ref{def:tau_N_blind} was inspired by  \cite{brown1981myopic}, who defined in a similar fashion the \textit{myopic topology} $\tau^m$, the finest myopic locally convex topology over $l^\infty$. This topology is generated by the family of myopic seminorms (see above Proposition \ref{prop:Nblind-seminorm}). A remarkable result of \cite{brown1981myopic} is that the myopic topology $\tau^m$ and the Mackey topology $\tau^{Mc}$ coincide, i.e. $\tau^{Mc}=\tau^m$.\footnote{Therefore they also coincide with the strict topology $\tau^s$, see Section \ref{sec:MathPrelim}.} This result is important since it justifies the use of the Mackey topology in general equilibrium. In his seminal paper, \cite{bewley1972existence} assumed that agents' preferences were continuous with respect to it to get summable prices.   The fact that the  Mackey topology coincides with the myopic topology behaviorally justifies this assumption. The interpretation is therefore that if agents discount the future an equilibrium exists. The converse was proved by \cite{araujo1985lack} who showed that myopia is actually necessary to get an equilibrium.

Corollaries \ref{cor:non-Haussdorf} and \ref{cor:nested} show that we obtained a nested sequence of non-Hausdorff topologies, all strictly contained in the product topology $\tau^p$. This result implies that when $N$ tends to infinity we do not achieve any kind of patience. Actually, since the product topology is strictly weaker than the myopic topology  $\tau^m$ of \cite{brown1981myopic}, an agent with $\tau^p$-continuous preferences will be at least as myopic as an agent with $\tau^m$-continuous preferences. 

 The next section makes precise in which sense an agent ``discounts" the future if her preferences are $\tau^p$-continuous.


\section{Eventual Blindness and the Product Topology}\label{sec:eventual-blind}

Section \ref{sec:eventual-blind} presents the main result of our paper, a behavioral characterization of the product topology. We saw in Corollary \ref{cor:nested} that $(\tau^{b}_N)_{N\in \N}$ forms a nested sequence of topologies all included in the product topology $\tau^p$. Here, we give a behavioral characterization of $\tau^p$-continuous preferences in terms of how agents discount the future. Consider the following definition.

\begin{defi}\label{def:EventualBlind}
A preference relation  $\pref$ on $l^{\infty}$ is eventually blind if for all ${\bf x }, {\bf y }\in l^{\infty}$ and for all $({\bf z}_k)_{k\in \N}$ with ${\bf z}_k \in l^{\infty}$, there exists $N\in \N$ such that if ${\bf x }\prefs {\bf y }$ then  ${\bf x } \prefs {\bf y } + T_n({\bf z}_n)$ for all $n\geq N$.
\end{defi}

Definition~\ref{def:EventualBlind} can be interpreted in the following way. Suppose that an agent strictly prefers  ${\bf x }$ to ${\bf y }$. Then there exists a time-period $N$ such that for all subsequent periods $n$ she will prefer stream ${\bf x }$ to stream ${\bf y }$, even if the latter is improved through $T_n({\bf z}_n)$. Note that, contrary to Definition~\ref{def:myopiaBL} of myopia, we are considering sequences of streams and not a fixed stream added to the tail of ${\bf y }$. Therefore, the decision maker keeps her preference for ${\bf x }$ over ${\bf y }$ when the streams are added late enough, even if one is ``compensating" the fact that the agent is getting the increase later and later in time. See Example~\ref{ex:Eventual_Blind}. It is not difficult to see on the one hand that eventual blindness implies myopia and on the other hand that $N$-blindness implies eventual blindness.

\begin{ex}\label{ex:Eventual_Blind}
In this example, we illustrate the concept of eventual blindness and we compare it with the one of myopia.

We start with myopia, introduced in Definition~\ref{def:myopiaBL}. To simplify the analysis consider three constant streams  ${\bf x }=(x,x,\dots)$, ${\bf y }=(y,y,\dots)$, ${\bf z }=(z,z,\dots)$ and suppose $x>y>0$ and $z>1$ such that $y+z>x$. If an agent has monotone preferences then  ${\bf x }\prefs {\bf y }$ and also  ${\bf y }+{\bf z }	\prefs {\bf x }$. Consider the sequence $({\bf y } + T_k({\bf z}))_{k\in \N}$ shown in the table below.

\begin{center}
\begin{tabular}{c|c|c|c|c|c|c|c|c}
k & 0 & 1 & 2& 3& \dots & n & n+1 & \dots \\ 
\hline 
0 & $y$ & $y+z$ & $y+z$  & $y+z$ & \dots & $y+z$ & $y+z$ &  \dots \\ 

1 & $y$ & $y$ & $y+z$ & $y+z$  & \dots & $y+z$ & $y+z$ &  \dots \\ 

2 & $y$ & $y$ & $y$ &$y+z$ & \dots & $y+z$ & $y+z$ &  \dots \\ 
 
\dots &  \dots & \dots &  \dots &  \dots &  \dots &  \dots &  \dots &  \dots \\ 
 
n & $y$ & $y$ & $y$ &$y$ & \dots & $y$ & $y+z$ &  \dots \\ 
\end{tabular} 
\end{center}
The first row corresponds to ${\bf y }+T_0({\bf z })$, increasing $k$ can be interpreted as postponing the date at which ${\bf y }$ will be improved adding the tail of ${\bf z}$.  The last row corresponds to	${\bf y } + T_n({\bf z})$, where we postponed this improvement until period $n+1$. According to   Definition~\ref{def:myopiaBL}, this new stream is strictly worse than ${\bf x }$ for a myopic decision maker whenever $n$ is big enough.

Consider now the sequence $({\bf z}_k)_{k\in N}$ where for all $k$, ${\bf z}_k=(z^k,z^k,\dots)\in l^\infty$. Let us now represent the sequence $({\bf y } + T_k({\bf z}_k))_{k\in \N}$ on a table.

\begin{center}
\begin{tabular}{c|c|c|c|c|c|c|c|c}
k & 0 & 1 & 2& 3& \dots & n & n+1 & \dots \\ 
\hline 
0 & $y$ & $y+z$ & $y+z$  & $y+z$ & \dots & $y+z$ & $y+z$ &  \dots \\ 

1 & $y$ & $y$ & $y+z^2$ & $y+z^2$  & \dots & $y+z^2$ & $y+z^2$ &  \dots \\ 

2 & $y$ & $y$ & $y$ &$y+z^3$ & \dots & $y+z^3$ & $y+z^3$ &  \dots \\ 
 
\dots &  \dots & \dots &  \dots &  \dots &  \dots &  \dots &  \dots &  \dots \\ 
 
n & $y$ & $y$ & $y$ &$y$ & \dots & $y$ & $y+z^n$ &  \dots \\ 
\end{tabular} 
\end{center}
As before, increasing $k$ means that we postpone the date at which ${\bf y }$ will be improved. However, in this case we are ``compensating" this delay by enlarging the amount assigned to the agent. This compensation is not contemplated in Definition~\ref{def:myopiaBL} for myopia. Therefore a myopic agent may very well have preferences ${\bf y } + T_k({\bf z}_k)\pref {\bf x}$ for all $k\in \N$. On the other hand, according to Definition~\ref{def:EventualBlind}, an eventually blind agent will have the opposite strict preference provided that $k$ is sufficiently large.

\end{ex}

In analogy with what we did in Section~\ref{sec:Nblind}, we define now the concept of eventually blind topologies.

 \begin{defi}
A locally convex topology $\tau$ is eventually blind if every $\tau$-continuous preference is eventually blind.  
\end{defi}

The following proposition characterizes eventually blind topology by giving a result on convergence of tails of sequences. 

\begin{propo}\label{prop:eventuallyblind-tails}
Let  $\tau$ be a locally convex topology on $l^{\infty}$. Then  $\tau$ is eventually blind  if and only if for all $({\bf x}_k)_{k\in \N}$, $T_n({\bf x}_n)\xrightarrow{\tau}_n {\bf 0}$.
\end{propo}
\begin{proof}
$\Rightarrow$ Again, we follow the ideas in the proof of Lemma 1 in \cite{brown1981myopic}. Suppose that $\tau$ is generated by the family of seminorms $\Gamma=\{ p_\alpha: \alpha\in A\}$. Then for all $\alpha\in A$ we can define a $\tau$-continuous preference $\pref_\alpha$ on $l^\infty$ by ${\bf x }\pref_\alpha {\bf y }$ if and only if $p_\alpha({\bf x })\geq p_\alpha({\bf y })$. \\
Fix now $p_\alpha\in\Gamma$. If $p_\alpha=0$ we are done. Otherwise there is $\bar{{\bf x }}$ such that $p_\alpha(\bar{{\bf x }})>0$.
Fix $\varepsilon>0$, $\alpha\in A$ and let $\beta=\frac{\varepsilon}{p_\alpha(\bar{{\bf x }})}$. Therefore $\varepsilon=p_\alpha(\beta \bar{{\bf x }})>0=p_\alpha({\bf 0})$ and $\beta\bar{{\bf x }}\prefs_\alpha {\bf 0}$. Since $\pref_\alpha$ satisfies eventual blindness of Definition~\ref{def:EventualBlind}, we have that $\beta\bar{{\bf x }}\prefs_\alpha T_n({\bf x}_n)$  when $n$ is sufficiently big, for all sequences $({\bf x}_k)_{k\in \N}$.  Then $\varepsilon=p_\alpha(\beta\bar{{\bf x }})>p_\alpha(T_n({\bf x}_n))$. Since this holds for all $\varepsilon>0$ and all $\alpha\in A$ we obtain $T_n({\bf x}_n)\xrightarrow{\tau}_n {\bf 0}$.

\medskip

$\Leftarrow$ Suppose by contradiction that ${\bf x }\prefs {\bf y }$ but there is a sequence $({\bf z}_k)_{k\in \N}$ such that  for all $k\in \N$ there is a $N_k$ such that ${\bf y } + T_{N_k}({\bf z}_{N_k})\pref {\bf x }$. Since $T_{N_k}({\bf z}_{N_k})\xrightarrow{\tau}_k {\bf 0}$, by continuity of $\pref$ we obtain ${\bf y }\pref {\bf x }$, a contradiction.

\end{proof}

As in Section~\ref{sec:Nblind}, we look for the finest eventually blind topology, since we want to have as many preferences as possible that are continuous with respect to it.

\begin{defi}\label{def:event-blind-topo}
We denote $\tau^b$ the finest eventually blind locally convex topology over $l^\infty$. The family of seminorms generating it is denoted $\Gamma^b$.
\end{defi}

The following proposition shows that  $\tau^b$ exists, and characterizes the family $\Gamma^b$ of eventually blind seminorms generating it.

\begin{propo}\label{prop:exist_T^b-and-seminorm}
$\tau^b$ exists and it is generated by the family $\Gamma^b=\{p_\alpha: \alpha\in A\}$ where for all $\alpha$ there exists $N$ such that $p_\alpha(T_n({\bf x}_n))=0$ for all $n\geq N$, for all $({\bf x}_k)_{k\in \N}$.
\end{propo}
\begin{proof}
Let $\tau$ be the topology generated by the family $\hat{\Gamma}^b$ of all seminorms $\{p_\alpha: \alpha\in A\}$,  such that $p_\alpha(T_n({\bf x}_n))\rightarrow_n 0$ for all sequences $({\bf x}_k)_{k\in \N}$. Clearly  $\varnothing \neq \Gamma^b_N\subset \Gamma^b\subseteq \hat{\Gamma}^b\;\forall N$, hence the locally convex topology $\tau^b$ exists and is eventually blind by Proposition~\ref{prop:eventuallyblind-tails}. Moreover, it is the finest since $\hat{\Gamma}^b$ contains all the seminorms with this property.\\
We show now that  $\Gamma^b=\hat{\Gamma}^b$. Suppose by contradiction that there exists $p_\alpha\in \hat{\Gamma}^b$ such that for all $k\in \N$ there is $n_k\geq k$ and ${\bf x}_{n_k}\in l^\infty$ such that $p_\alpha(T_{n_k}({\bf x }_{n_k}))>0$. Define the new sequence of streams $\bar{{\bf x }}_n$
$$
\bar{{\bf x }}_n=\begin{cases}
{\bf 0} & \text{ if } n\neq n_k \\
\frac{{\bf x }_{n_k}}{p_\alpha(T_{n_k}({\bf x }_{n_k}))} &\text{ otherwise.}
\end{cases}
$$
Note that for all $n$, $\bar{{\bf x }}_n\in l^\infty$. However,  $p_\alpha(T_{n_k}(\bar{{\bf x }}_{n_k}))=1$ for all $k$ and therefore $T_n(\bar{\bf x}_n)\not\xrightarrow{\tau}_n {\bf 0}$, contradicting Proposition~\ref{prop:eventuallyblind-tails}.
\end{proof}

Coherently with the nomenclatures given above, we call \textit{eventually blind seminorms} the ones with the property given in Proposition \ref{prop:exist_T^b-and-seminorm}.

\begin{remark}\label{rmk:T^b_subst_T^m}
Note that for all $N\in \N$, $N$-blind seminorms are eventually blind seminorms. Moreover,  eventually blind seminorms are also myopic seminorms (see the discussion above Proposition~\ref{prop:Nblind-seminorm}). Therefore for all $N\in \N$
$$
\tau^b_N\subsetneq\tau^b\subseteq \tau^m.
$$
\end{remark}

We are now ready to present the main result of the paper. Proposition~\ref{prop:Tb=Tp} below characterizes the product topology  $\tau^p$, showing that it is the finest eventually blind locally convex topology over $l^\infty$.

\begin{propo}\label{prop:Tb=Tp}
$\tau^b=\tau^p$.
\end{propo}
\begin{proof}
The proof consists in a sequence of lemmata.

\begin{lema}\label{lem:1}
 $\tau^b\supseteq \tau^p$.
\end{lema}
\begin{proof}
  Note that for all $({\bf x}_k)_{k\in \N}$, $T_n({\bf x}_n)\xrightarrow{\tau^p}_n {\bf 0}$ and therefore $\tau^p$ is an eventually blind topology by Proposition~\ref{prop:eventuallyblind-tails}. Since $\tau^b$ is the finest eventually blind topology, we obtain $\tau^b\supseteq \tau^p$.  
\end{proof}

\begin{lema}\label{lem:dual_ev_blind}
 The dual of $l^\infty$ w.r.t. $\tau^b$ is $c_{00}$, i.e. $(l^\infty,\tau^b)'=c_{00}$.
\end{lema}
\begin{proof}
This is proved in Proposition~\ref{prop:dual_T^b}, Section~\ref{sec:dual_applications}.
\end{proof}

\begin{lema}\label{lem:dual_prod}
$\tau^p$ is the finest topology such that its dual is $c_{00}$.
\end{lema}
\begin{proof}
Consider a locally convex topology $\tau$, strictly finer than $\tau^p$, i.e. $\tau\supsetneq \tau^p$. Then there exists a set of indices $J\subseteq\N$ with infinite cardinality so that a neighborhood $\mathcal{N}$ of ${\bf 0}$ is given by $\mathcal{N}=\prod_{n\in J}(-\delta_n,\delta_n)\times \prod_{n\not\in J} \mR$, where $\delta_n>0\;\forall n\in J$. Take a sequence ${\bf b}\in l^1$ such that ${b}_k=0$ if $k\not\in J$ and let $\sum_{n\in J}|b_n|=M$.\\
Define now the  function $f:l^\infty\rightarrow\mR$ by $f({\bf x})=\sum_n x_n b_n$ for all ${\bf x}\in l^\infty$. Clearly, $f$ is well defined and linear. We prove now that $f$ is continuous at ${\bf 0}$. Fix $\epsilon>0$ and take $\mathcal{O}\in \tau$, an open neighborhood of ${\bf 0}$,  defined by $\mathcal{O}=\{{\bf z}\in l^\infty : |z_k|<\frac{\epsilon}{M} \text{ if } k\in J\}$. Then for all ${\bf z}\in \mathcal{O}$ we have
$$
|f({\bf z})-f({\bf 0})| = |\sum_n z_n b_n| 
    \leq \sum_{n} |z_n b_n|
    =\sum_{n\in J} |z_n|\cdot |b_n|
    \leq\frac{\epsilon}{M} \sum_{n\in J}  |b_n|=\epsilon.
$$
This implies that $f$ is continuous at ${\bf 0}$ and hence on all $l^\infty$. \\
Since ${\bf b}\not\in c_{00}$, this implies that the dual space of $\tau$ is such that $(l^\infty,\tau)'\supsetneq c_{00}$. Hence $\tau^p$ is the finest topology such that its dual is $c_{00}$.
\end{proof}

\begin{lema}\label{lem:4}
 $\tau^b\subseteq \tau^p$.
\end{lema}
\begin{proof}
By Lemma~\ref{lem:dual_ev_blind}, $(l^\infty,\tau^b)'=c_{00}$. By Lemma~\ref{lem:dual_prod},  $\tau^p$ is the finest topology with dual $c_{00}$. Therefore  $\tau^b\subseteq \tau^p$.
\end{proof}

\noindent Lemma~\ref{lem:1} and Lemma~\ref{lem:4} imply that $\tau^b=\tau^p$.
\end{proof}

While continuity with respect to the Mackey topology was behaviorally characterized by the concept of myopia in \cite{brown1981myopic}, the product topology received much less attention. The only paper that investigates possible characterization of the product topology that we are aware of is the one of \cite{streufert1993abstract}, who studied in an abstract framework a concept called tail insensitivity.  Both product and Mackey topologies have been fruitfully used in the literature of general equilibrium in the seminal paper of \cite{peleg1970markets} (product topology) and  
\cite{bewley1972existence} (Mackey topology). Proposition~\ref{prop:Tb=Tp} fills a gap in the literature as it gives an explicit and meaningful economic interpretation to preferences that are continuous with respect to the product topology. 

We conclude this section with some examples of eventually blind utility functions, and we study their continuity properties.

\begin{ex}\label{ex:utilities}
Consider the following utility functions.
\begin{enumerate}
    \item Let $v:\mR\rightarrow\mR$ a continuous function. Define the utility function $u:l^\infty \to \mR$ by $u({\bf x})= \sum_{k=0}^N v(x_k)$. Then $u$ is $N$-blind, and it is easy to see that it is continuous with respect to $\tau^b_N$.
    
    \item Let  $v:\mR\rightarrow \mR$, $v(x):=\frac{x}{1+|x|}$, and define $u:l^\infty \to \mR$, by  $u({\bf x}):=\sum_{n=0}^\infty\frac{v(x_n)}{2^{n+1}}$. This utility function is eventually blind and $\tau^b$-continuous. However, it is neither $N$-blind nor $\tau^b_N$-continuous.
    
    Function  $u$ is well-defined since  $\sum_n\left|\frac{v(x_n)}{2^{n+1}}\right|=\sum_{n=0}^\infty\frac{|v(x_n)|}{2^{n+1}}\leq \sum_{n=0}^\infty\frac{1}{2^{n+1}}=1$.  
    
    We prove continuity with respect to $\tau^b$. Since $v$ is continuous, $f_n(x)=\frac{v(x)}{2^{n+1}}$ is a continuous real function. Defining $g_n:l^\infty\to \mR$ as $g_n=f_n\,\circ p_n$ where $p_n$ is the $n$-th projection, we immediately get that $g_n$ is a $\tau^b$-continuous function. For any $m\in \mN$ the function $u_m({\bf x}):=\sum_{n=0}^m\frac{v(x_n)}{2^{n+1}}=\sum_{n=0}^mg_n({\bf x})$ is a sum of  $\tau^b$-continuous functions and therefore is $\tau^b$-continuous. Finally , $\forall {\bf x}\in l^\infty$, $|u({\bf x})-u_m({\bf x})|=\left|\sum_{n=m+1}^\infty\frac{v(x_n)}{2^{n+1}} \right|\leq \sum_{n=m+1}^\infty\left|\frac{v(x_n)}{2^{n+1}} \right|\leq \sum_{n=m+1}^\infty\frac{1}{2^{n+1}}=\frac{1}{2^{m+1}}\to_m 0$, hence the sequence of $\tau^b$-continuous functionals $(u_m)_m$ converges uniformly to $u$, which implies that $u$ is also $\tau^b$-continuous. By definition of $\tau^b$, the utility function is eventually blind.

    We  show that $u$ is not $\tau^b_N$-continuous at ${\bf 0}$. Consider the constant sequence ${\bf 1}=(1,1,\dots)$ and note that $(i)$ $u(T_{N+1}({\bf 1}))=\frac{1}{2^{N+2}}$; $(ii)$ $T_{N+1}({\bf 1})$ is in every open neighborhood of ${\bf 0}$. Let $\varepsilon< \frac{1}{2^{N+2}}$ and consider the open neighborhood $(-\varepsilon,\varepsilon)$ of $u({\bf 0})=0$. If $u$ were to be $\tau^b_N$-continuous then there would exist an open neighborhood $V$ of ${\bf 0}$ such that $u(V)\subset (-\varepsilon,\varepsilon)$. But every such $V$ contains $T_{N+1}({\bf 1})$ by $(ii)$, while  $u(T_{N+1}({\bf 1}))>\varepsilon$ by $(i)$. Hence, no such $V$ exists and $u$ is discontinuous at ${\bf 0}$.

    \item $u({\bf x})=\sum_n\frac{x_n}{2^{n+1}}$ is $\tau^{Mc}$-continuous and $\tau^b$-discontinuous.

    Clearly, $u$ is well-defined since  $\sum_n\left|\frac{x_n}{2^{n+1}}\right|\leq \|{\bf x}\|_\infty\sum_n\frac{1}{2^{n+1}}=\|{\bf x}\|_\infty<+\infty$. Note that sequence $(2^{-(n+1)})_n\in l^1$. It is well known that every element of ${\bf a}\in l^1$ defines a linear $\tau^{Mc}$-continuous function $f_{{\bf a}}$ over $l^\infty$ by $f_{{\bf a}}({\bf x})=\sum_n x_na_n$. Since $u$ is defined precisely in this way, it is $\tau^{Mc}$-continuous. To see that $u$ is not $\tau^b$-continuous, take the sequence defined for all $k\in \N$ by ${\bf z}_k=T_k(2^{k+1}{\bf 1})$, where ${\bf 1}$ is the constant sequence equal to one. Clearly  ${\bf z}_k\xrightarrow{\tau^b}_n {\bf 0}$. However, $u({\bf z}_k)=\sum_{n=k+1}^\infty\frac{2^{k+1}}{2^{n+1}}=2^{k+1}\sum_{n=k+1}^\infty\frac{1}{2^{n+1}}=2^{k+1}\frac{1}{2^{k+1}}=1$, for all $k\in\N$. Therefore ${u}({\bf z}_k)\not\rightarrow_k {u}({\bf 0})=0$ and $u$ is $\tau^b$-discontinuous.

    \item   Let $\succeq_l$ be the intertemporal lexicographic preference, that is, that ${\bf x}\succ_l {\bf y}\LR \exists N\in \mN$ such that $x_N>y_N$ and $x_n\geq y_n\;\forall n\leq N-1$ and ${\bf x}\sim_l {\bf y} \LR {\bf x}={\bf y}$. Then $\succeq_l$ is eventually blind, but it is not $\tau^\infty$-continuous (and hence not $\tau^b$-continuous).

    Suppose ${\bf x}\succ_l {\bf y}$, then there is $N$ such that  $x_N>y_N$ and $x_n\geq y_n\;\forall n\leq N-1$. Hence for all sequences $({\bf z}_k)_{k\in \N}$ with ${\bf z}_k \in l^{\infty}$,   ${\bf x } \prefs {\bf y } + T_n({\bf z}_n)$ for all $n\geq N$. 

    We show that $\succeq_l$ is not  $\tau^\infty$-continuous. Consider the sequence ${\bf z}_k=\left(\frac{1}{k},-1,-1,\dots \right)$. Then for all $k\in \N$, ${\bf z}_k\succ_l {\bf 0}$ and moreover ${\bf z}_k\xrightarrow{\tau^\infty}_k \left(0,-1,-1,\dots \right)$. However, ${\bf 0}\succ_l \left(0,-1,-1,\dots \right)$, hence $\succeq_l$ is not  $\tau^\infty$-continuous.

\end{enumerate}
\end{ex}


\section{Dual Spaces}\label{sec:dual_applications}

The objective of this section is to study the dual spaces of the $N$-blind topology $\tau^b_N$, and the eventually blind topology $\tau^b$. Dual spaces are important objects as continuous linear functionals generalize the notion of prices when the space of commodities is infinite dimensional. Note that in Section~\ref{sec:eventual-blind}, Proposition~\ref{prop:Tb=Tp}, we proved that the $\tau^b=\tau^p$, and it is well known that the dual of the product topology is $c_{00}$. However, to prove Proposition \ref{prop:Tb=Tp} we claimed in Lemma~\ref{lem:dual_ev_blind} that $(l^\infty,\tau^b)'=c_{00}$, which we still need to show.

We begin by characterizing all linear functions which are $\tau^{b}_N$-continuous (we recall that the sets $c_{00}^N$ and $c_{00}$ appearing in the propositions below have been defined in Section~\ref{sec:MathPrelim}).

\begin{propo}\label{prop:dual_T^b_N}
The following assertions are equivalent: 
\begin{enumerate}
    \item[(i)] The linear functional $f:l^{\infty}\to \mathbb{R}$ is $\tau^{b}_N$-continuous.
    \item[(ii)] $f$ is a linear functional and $f(T_N)=\{0\}$,  where $T_N=\{ T_N({\bf x}):{\bf x}\in l^\infty\}$.
    \item[(iii)] There exists $ {\bf a}\in c_{00}^N$ such that $f({\bf x})=\sum_{n=0}^\infty a_nx_n$.
\end{enumerate}
\end{propo}

\begin{proof}
$(i)\R (ii)$. By Corollary~\ref{cor:non-Haussdorf} we have that $T_N=\overline{\{{\bf 0}\}}$. Therefore, since 
$f$ is $\tau^{b}_N$-continuous, $|f(T_N({\bf x}))-0|<\epsilon$ for all $\epsilon>0$ and for all ${\bf x}\in l^\infty$. Hence  $f(T_N({\bf x}))=0$ for all ${\bf x}\in l^\infty$.

\medskip

$(ii)\R (iii)$. Let $a_n:=f({\bf 1}_n)$ for each $n\in \mN$. For all ${\bf x}\in l^\infty$ we have
\begin{align*}
    f({\bf x})&=f(H_N({\bf x})+T_N({\bf x}))\\
    &= f(H_N({\bf x}))+f(T_N({\bf x}))\\
    &= f(H_N({\bf x})),
\end{align*}
where the second equality follows from linearity of $f$ and third one follows from  $f(T_N)=\{0\}$. Again by linearity for all ${\bf x}\in l^\infty$, $f({\bf x})=f(H_N({\bf x}))=\sum_{n=0}^N a_nx_n$. Hence ${\bf a}\in c_{00}^N$.

\medskip

$(iii)\R (i)$. Let ${\bf a}\in c_{00}^N$ and $f({\bf x})=\sum_{n=0}^N a_nx_n$. Since $f$ is clearly linear, we only need to prove that it is continuous at ${\bf 0}$. By Proposition~\ref{prop:N-blind_topo}, $\tau^{b}_N$ is generated by the family of seminorms $P^N=\{p_k:k=0,\dots, N\}$, where $p_k({\bf x })=|x_k|$. Fix $\epsilon>0$ and consider the neighborhood $\mathcal{N}$ of ${\bf 0}$ defined by 
$$
\mathcal{N}=\left\{ {\bf x }\in l^\infty : |x_k|<\frac{\epsilon}{(N+1) \max_i a_i}, k=0,\dots, N \right\}.
$$
Then for all ${\bf x} \in \mathcal{N}$ one has
$$
|f({\bf x})-0|\leq \sum_{i=0}^N |a_ix_i|\leq (\max_{i} a_i) \sum _{i=0}^N|x_i|<(\max_{i} a_i)  (N+1) \frac{\epsilon}{(N+1) \max_i a_i}=\epsilon.
$$
Hence $f$ is $\tau^{b}_N$-continuous.
\end{proof}

When agents have preferences that are continuous with respect to $\tau^{b}_N$ they are future-blind and they treat the infinite dimensional space $l^\infty$ as a space with dimension $N+1$. Hence, as shown in Proposition~\ref{prop:dual_T^b_N}, also the dual space will be finite dimensional and can be identified with $c_{00}^N$.

We now study the dual space of $\tau^b$.

\begin{propo}\label{prop:dual_T^b} The following assertions are equivalent:
\begin{enumerate}
    \item[(i)] The linear functional $f$ is $\tau^{b}$-continuous.
    \item[(ii)]  $f$ is a linear functional and there exists $N$ such that $f(T_N)=\{0\}$. 
    \item[(iii)] There exists $ {\bf a}\in c_{00}$ such that $f({\bf x})=\sum_{n=0}^\infty a_nx_n$.
\end{enumerate}
\end{propo}

\begin{proof} 
We first prove an auxiliary lemma.
\begin{lema} \label{lem:aux} Let ${\bf b}\in l^1$ be such that $b_n\neq 0$ for all $n\in J$, where $J$ is an infinite subset of $\mN$. Then $J'=\{k\in J:\sum^\infty_{n=k}b_n\neq 0\}$ is also an infinite set.  
\end{lema}

\begin{proof} 
Suppose $J'$ is finite, take $N\in J$ such that $N>\max \{J'\}$ and define $J_N:=\{k\in J:k\geq N \}$, then $\sum^{\infty}_{n=k}b_n=0\,\,\forall k\in J_N$. But, given any two consecutive numbers $k_1,k_2\in J_N$, such that, say, $k_2>k_1$, one has 
$$
0=0-0=\sum^\infty_{n=k_2}b_n-\sum^\infty_{n=k_1}b_n=\sum^{k_2-1}_{n=k_1}b_n=b_{k_1}.
$$ 
The last equality comes form the fact that $k_1+1,...,k_2-1 \notin J\R b_{k_1+1}=...=b_{k_2-1}=0$. But $b_{k_1}= 0$ contradicts the definition of $J$. Therefore, $J'$ is also an infinite subset of $\mN$. 
\end{proof}

$(i)\Rightarrow (ii)$. We prove it by contradiction. As noted in Remark~\ref{rmk:T^b_subst_T^m} we have $\tau^b \subseteq \tau^m$. This implies $(l^\infty,\tau^{b})'\subseteq (l^\infty,\tau^{m})' = l^1$, where the last equality is shown in  \cite{brown1981myopic}. Therefore all continuous linear functions in $(l^\infty,\tau^{b})'$ can be written as $f({\bf x})=\sum^{\infty}_{n=0}b_n x_n$ for some ${\bf b}\in l^1$. Suppose by contradiction that we have  ${\bf b}\in l^1\setminus c_{00}$. For such ${\bf b}$ we can find two infinite sets  $J$ and $J'$  as in Lemma~\ref{lem:aux}. For all $k\in J'$, define the constant sequence ${\bf c}_k(n)= \frac{1}{b_k\cdot \sum^\infty_{i=k}b_i}$ for all $n\in \N$. Take now the sequence  $({\bf x}_k)_{k\in \mN}$ of elements in $l^\infty$ defined by 
$$
{\bf x }_k=\begin{cases}
{\bf 0} & \text{ if } k=0\;\textrm{or} \;k\not\in J' \\
T_{k-1}({\bf c}_k) &\text{ if } k\in J'
\end{cases}
$$
and note that for $k$ in $J'$ one has
$$
f({\bf x}_k)= \sum^\infty_{i=k}b_i \left(\frac{1}{b_k\cdot \sum^\infty_{i=k}b_i} \right)=\frac{1}{b_k}\frac{ \sum^\infty_{i=k}b_i}{ \sum^\infty_{i=k}b_i}=\frac{1}{b_k}.
$$
By  Lemma~\ref{lem:aux}, $J'$ is an infinite set, therefore $(f({\bf x}_k))_{k\in J'}$ is a subsequence of $(f({\bf x}_k))_{k\in \mN}$ and because ${\bf b}\in l^1$, then $\lim_{k\to\infty} b_k = 0$. Moreover, by Proposition~\ref{prop:eventuallyblind-tails} we have that ${\bf x}_k=T_{k-1}({\bf x}_k)\xrightarrow{\tau^b}_k {\bf 0}$. Thus,
$$\lim_{\substack{k\in J' \\ k\to \infty}} b_k=0 \R \lim_{\substack{k\in J' \\ k\to \infty}} f({\bf x}_k)=\lim_{\substack{k\in J' \\ k\to \infty}}\frac{1}{b_k}\neq  0 \R \lim_{k\to  \infty}  f({\bf x}_k)\neq 0=f({\bf 0}).$$ 
Hence, $f$ is not $\tau^{b}-$continuous. 

\medskip

$(iii)\R(ii)$. Let ${\bf a}\in c_{00}$ and $f({\bf x})=\sum^\infty_{n=0}a_n x_n$. Since there exists $N$ such that $a_n=0$ for all $n\geq N$, for all ${\bf x}\in l^\infty$
$$
f(T_N({\bf x}))=\sum^N_{n=0} a_n 0+\sum^\infty_{n=N+1} 0 x_n=0,
$$
Therefore there exists $N$ such that $f(T_N)=\{0\}$.

\medskip

$(ii)\R(i)$. If $f$ is a linear functional such that $f(T_N)=\{0\}$ for some $N\in \mN$, then by Proposition \ref{prop:dual_T^b_N}, $f$ is $\tau^{b}_N$-continuous, and consequently is $\tau^{b}$-continuous. 

\medskip

$(ii)\R(iii)$. This can be shown as in the proof of Proposition~\ref{prop:dual_T^b_N}.

\end{proof}




\section{Conclusions}\label{sec:conclusions}

In this article, we studied topologies that are consistent with future-blind agents. More specifically, we introduced two types of behaviors. On the one hand, $N$-blind decision makers completely neglect the future after period $N$, on the other hand, eventually blind agents overlook future periods starting from some point in time, even if some compensation is given to them. 

When the space is infinite dimensional, continuity of preferences entails strong behavioral properties. The paper's main contribution is to study locally convex topologies coherent with  $N$-blind and eventually blind behaviors, as was done in the seminal paper of \cite{brown1981myopic}. We showed that $N$-blindness is characterized by continuity with respect to the non-Hausdorff topology $\tau^b_N$, while eventual blindness characterizes the product topology $\tau^p$. To the best of our knowledge, we are the first to characterize how continuity with respect to the product topology translates into time discounting.

Finally, we studied the spaces of continuous linear functionals, interpreted as a generalized version of price vectors. When paired with $\tau^b_N$, the dual space of $l^\infty$ is equal to $c^N_{00}$, the set of sequences that are 0 after period $N$. When paired with $\tau^b$ is equal to $c_{00}$, the set of sequences with finitely many non-zero elements.

We conclude with Table \ref{tab:conclusion}, which summarizes our main results and compares them with the ones obtained by \cite{brown1981myopic}.

\begin{table}[h]
    \centering
    \begin{tabular}{c|c|c|c}
        Behavior & Topology & Tail Behavior & Dual Space \\
        \hline
        $N$-blindness & $\tau^b_N$ & $T_N=\overline{\{\bf 0\}}$  & $c_{00}^N$ \\
         \hline
        Eventual Blindness & $\tau^b=\tau^p$ &   $T_n({\bf x}_n)\xrightarrow{\tau^b}_n {\bf 0}$, $\forall ({\bf x}_k)_k$ & $c_{00}$ \\
         \hline
        Myopia & $\tau^m=\tau^{Mc}$  & $T_n({\bf x})\xrightarrow{\tau^m}_n {\bf 0}$, $\forall {\bf x}$   & $l^1$ \\
    \end{tabular}
    \caption{Behaviors and topologies.}
    \label{tab:conclusion}
\end{table}

\newpage


\bibliographystyle{plainnat}
\bibliography{references}

\end{document}